\documentclass[conference,10pt]{IEEEtran}
\IEEEoverridecommandlockouts
\usepackage{cite}
\usepackage{amsmath,amssymb,amsfonts}
\usepackage{graphicx}
\usepackage{textcomp}
\usepackage{xcolor}
\usepackage{amsthm}
\usepackage{graphics}
\usepackage{multicol}
\usepackage{lipsum}
\usepackage{color}
\usepackage{mathrsfs}
\usepackage{mathtools}
\usepackage{amsbsy}
\usepackage[colorlinks=true,bookmarks=false,citecolor=blue,urlcolor=blue]{hyperref}
\usepackage{xcolor}
\usepackage{mathrsfs}
\usepackage{setspace}
\usepackage{float}
\usepackage{graphicx}
\usepackage{pstool}
\usepackage{cite}
\usepackage{lettrine}
\usepackage[normalem]{ulem}
\usepackage{latexsym}
\usepackage{algpseudocode}
\usepackage{algorithm,algpseudocode}
\usepackage{algorithmicx}
\usepackage{multirow}
\usepackage{wasysym}
\usepackage{cite}
\usepackage{mathrsfs}
\usepackage{amsmath,cite,amsfonts,amssymb,color}

\ifCLASSOPTIONcompsoc
\usepackage[caption=false,font=normalsize,labelfon
t=sf,textfont=sf]{subfig}
\else
\usepackage[caption=false,font=footnotesize]{subfig}
\fi
\allowdisplaybreaks
\newtheorem{remark}{Remark}
\newtheorem{theorem}{Theorem}
\newtheorem{lemma}{Lemma}
\newcommand{\imj}{\mathrm{j}}

\graphicspath{{figures/}}
\def\BibTeX{{\rm B\kern-.05em{\sc i\kern-.025em b}\kern-.08em
    T\kern-.1667em\lower.7ex\hbox{E}\kern-.125emX}}
\begin{document}
\title{Analysis of Orthogonal Matching Pursuit for Compressed Sensing in Practical Settings
}

\author{\IEEEauthorblockN{Hamed Masoumi, Michel Verhaegen and Nitin Jonathan Myers}
\IEEEauthorblockA{Delft Center for Systems and Control, Delft University of Technology, The Netherlands\\
Email: \{H.Masoumi, M.Verhaegen, N.J.Myers\}@tudelft.nl}
}

\maketitle
\begin{abstract}
Orthogonal matching pursuit (OMP) is a widely used greedy algorithm for sparse signal recovery in compressed sensing (CS). Prior work on OMP, however, has only provided reconstruction guarantees under the assumption that the columns of the CS matrix have equal norms, which is unrealistic in many practical CS applications due to hardware constraints. In this paper, we derive sparse recovery guarantees with OMP, when the CS matrix has unequal column norms. Finally, we show that CS matrices whose column norms are comparable achieve tight guarantees for the successful recovery of the support of a sparse signal and a low mean squared error in the estimate. 
\end{abstract}

\begin{IEEEkeywords}
Compressive sensing, orthogonal matching pursuit, support recovery, mutual coherence.
\end{IEEEkeywords}

\section{Introduction}\label{sec1}
\par CS is a method to recover a sparse signal from its compressed representation \cite{candes2008introduction}, which is acquired by multiplying a known CS matrix with the unknown sparse signal. Reconstructing this signal is an ill-posed problem, as CS matrices have fewer rows than columns. CS algorithms, however, can estimate the unknown signal by exploiting a sparse prior.
\par The OMP \cite{tropp2007signal} is a popular greedy algorithm in CS. A key challenge in CS is to find the conditions under which a unique sparse solution is guaranteed. Such conditions are usually based on the restricted isometry constants (RICs) \cite{cai2011orthogonal,wang2012near,chang2014improved,liu2017some} and the \textit{coherence} of the CS matrix \cite{ben2010coherence,chi2012coherence,miandji2017probability,emadi2018performance,stankovic2021improved,amiraz2021tight}. Conditions based on the RICs depend on the restricted isometry property (RIP) of the CS matrix, which can be obtained in a tractable manner for certain random constructions. Verifying whether a given deterministic matrix satisfies the RIP, however, is NP-hard in general \cite{bandeira2013certifying}. The coherence of a CS matrix, i.e., the maximum absolute inner product of its two distinct \textit{$\ell_2$-normalized} columns, is tractable to compute for a given CS matrix. As a result, it is convenient to derive coherence-based CS guarantees when a CS matrix is specified.
\par In several applications, the columns of the CS matrix have unequal $\ell_2$ norms, such as in the case of CS-based wireless channel estimation with low-resolution phased arrays. The hardware imperfections in such arrays result in inevitable variations in the column norms of the CS matrix. Although CS algorithms, such as the OMP, are known to perform well even with such matrices, the impact of the variation in the column norms on sparse signal recovery remains unclear. To address this gap in the literature, we study the performance guarantees of the OMP algorithm for generic CS matrices whose columns can have different norms. 
\par Now, we discuss the literature on coherence-based performance guarantees of the OMP algorithm. For a noiseless CS setting, \cite{ stankovic2021improved } provides a guarantee for successful support recovery and exact reconstruction of the sparse signal. In \cite{ben2010coherence,chi2012coherence}, sufficient conditions for successful support recovery are derived and an upper-bound for the mean square error (MSE) of the recovered signal is presented. Finally, assuming that the CS measurements are corrupted by additive Gaussian noise, \cite{miandji2017probability,amiraz2021tight} obtain improved coherence-based performance guarantees for the OMP. The works in \cite{stankovic2021improved,ben2010coherence,chi2012coherence,miandji2017probability,amiraz2021tight}, assume a column-normalized CS matrix, i.e., its columns have the same $\ell_2$ norm. When the CS matrix is not column-normalized, \cite{bruckstein2008sparse} provides coherence-based conditions to recover the support of a sparse vector with \textit{non-negative} entries from a \textit{noiseless} set of measurements using the OMP. 

In this paper, we derive sufficient conditions to identify the support of a complex-valued sparse vector when the columns of the CS matrix have unequal $\ell_2$-norms and when the measurements are noisy. Then, we derive an upper-bound for the MSE of the reconstructed signal with the OMP. Our derivation follows a similar structure as the one in \cite{ben2010coherence}. Hence, our results reduce to those in \cite{ben2010coherence} when the CS matrix has equal column norms. For a fixed Frobenius norm on the CS matrix, we show that CS matrices with similar column norms achieve better sparse recovery performance using the OMP compared to those with substantially different column norms. 

\emph{Notation}: $a$, $\mathbf{a}$ and $\mathbf{A}$ denote a scalar, vector, and a matrix. $a_i$ is the $i^{\mathrm{th}}$ entry of $\mathbf{a}$. Also, $\mathbf{a}_i$ is the $i^{\mathrm{th}}$ column of $\mathbf{A}$. $\Vert\cdot\Vert_2$ is the $\ell_2$-norm operator.
$(\cdot)^{T}$ and $(\cdot)^{\ast}$ denote the transpose and conjugate-transpose operators. $[N]$ denotes the set $\{1,2,...,N\}$. $\mathcal{CN}(0,\sigma^{2})$ is the zero-mean complex Gaussian distribution with variance $\sigma^{2}$. $\mathbf{0}_{N\times 1}$ denotes an $N\times 1$ all-zero vector. $\mathbb{E}[\cdot]$ is the expectation operator and $\mathrm{Pr}\{ \mathcal{E} \}$ denotes probability of an event $\mathcal{E}$. Finally, $\mathsf{j}=\sqrt{-1}$. 
\section{Preliminaries}\label{sec2}
Consider a $k$-sparse $N\times 1$ complex vector $\mathbf{x}$ defined over a support set $\Lambda\subseteq [N]$. The cardinality of the support set is defined by $|\Lambda|=k$. We use $\mathbf{A}$ to denote an $M\times N$ CS matrix with $M<N$ and $\mathbf{y}$ to denote the $M$ noisy CS measurements. The noise in the measurements is modeled as an 
 $M\times 1$ Gaussian random vector $\mathbf{v}$ whose entries are independently distributed as  $\mathcal{CN}(0,\sigma^{2})$. Thus, 
\begin{equation}\label{eqn:measurements}
    \mathbf{y}=\mathbf{A}\mathbf{x}+\mathbf{v}.
\end{equation}
For the $j^{\mathrm{th}}$ column of $\mathbf{A}$, we define its column norm as $d_j=\|\mathbf{a}_j\|_2$. We also define $d_{\mathrm{max}} = \underset{j\in[N]}{\max}~d_{j}$ and $d_{\mathrm{min}} = \underset{j\in[N]}{\min}~d_{j}$. We assume that the Frobenius norm of the CS matrix is fixed to $\sqrt{N}$, equivalently $\sum_{j=1}^{N}d_j^2=N$. Under this assumption, \begin{equation}\label{eqn:dmin}
0 < d_{\min}\leq 1,
\end{equation}
\begin{equation}\label{eqn:dmax}
  1 \leq d_{\max} < \sqrt{N}.
\end{equation}
When $d_\mathrm{min}=d_\mathrm{max}$, we observe that $\mathbf{A}$ becomes a column-normalized matrix, i.e., a matrix with equal column norms. 
\par The mutual coherence $\mu$ of $\mathbf{A}$ is a measure of its quality for sparse recovery and it is defined as
\begin{equation}\label{eqn:coherence}
    \mu = \underset{\{(j, \ell): j\neq \ell, j\in[N], \ell\in[N]\}}{\max} ~\frac{|\mathbf{a}_{j}^{*}\mathbf{a}_{\ell}|}{d_{j}d_{\ell}}.
\end{equation}
A CS matrix with a small $\mu$ results in better support recovery and in a smaller error in the estimate of the sparse vector \cite{ben2010coherence}.
\par The OMP can estimate $\mathbf{x}$ from the CS measurements $\mathbf{y}$ in \eqref{eqn:measurements}, with a known sparsity level $k$, in $k$ iterations \cite{ben2010coherence}. In each iteration of the OMP, one column of $\mathbf{A}$ is selected and an estimate of the entries of the sparse vector corresponding to the selected columns is obtained. Let $\Lambda^{i}$ denote the estimated support and $\hat{\mathbf{x}}^i$ denote the sparse vector estimated in the $i^{\mathrm{th}}$ iteration of the OMP algorithm. Also, let $\hat{\mathbf{x}}^i_{\Lambda^{i}}$ be a subvector of $\hat{\mathbf{x}}^i$ indexed by $\Lambda^{i}$ and $\mathbf{A}_{\Lambda^{i}}$ be a matrix obtained by retaining columns of $\mathbf{A}$ indexed by $\Lambda^{i}$. The OMP algorithm is summarized in Algorithm \ref{alg:omp} when the CS matrix has different column norms \cite{gharavi1998fast}.
\begin{algorithm}[H]
\small
		\caption{OMP algorithm from \cite{gharavi1998fast}.} \label{alg:omp}
		\textbf{Input}: Sparsity level $k$, CS measurements $\mathbf{y}$ and CS matrix $\mathbf{A}$. 

  \textbf{Initialization}: $i=1$, $\Lambda^{0}=\emptyset$, $\hat{\mathbf{x}}^{0}=\mathbf{0}_{N\times 1}$ and $\mathbf{r}^{0}=\mathbf{y}$.

  \textbf{While} $i\leq k$ \textbf{do}:
	\begin{itemize}
		\item[\textbf{1}.] $s         =\underset{j\in[N]}{\mathrm{argmax}}   ~\dfrac{|\mathbf{a}_{j}^{*}\mathbf{r}^{i-1}|}{d_j}.
         $
		\item[\textbf{2}.] $\Lambda^i = \Lambda^{i-1}\cup s.$
            \item[\textbf{3}.] $\hat{\mathbf{x}}^{i}_{\Lambda^{i}} = (\mathbf{A}_{\Lambda^{i}}^{*}\mathbf{A}_{\Lambda^{i}})^{-1}\mathbf{A}_{\Lambda^{i}}^{*}\mathbf{y}.$
            \item[\textbf{4}.] $\mathbf{r}^{i} = \mathbf{y}-\mathbf{A}\hat{\mathbf{x}}^{i}.$
            \item[\textbf{5}.] $i\leftarrow  i + 1$.
		\end{itemize}
  \textbf{End While}
  
		\textbf{Output}: $\hat{\mathbf{x}}^{i}$.
\end{algorithm}
The key steps in Algorithm \ref{alg:omp} are \textit{i)} correctly detecting the support of $\mathbf{x}$ in step \textbf{1} and \textit{ii)} estimating the entries of $\mathbf{x}$, detected in step \textbf{1}, by solving a least-squares problem in step \textbf{3}. Hence, we closely analyze these two steps in Section \ref{sec3} and provide conditions that guarantee support recovery with high probability in step \textbf{1}. We also derive an upper bound on the MSE of the sparse vector estimated in step \textbf{3}.
\section{Coherence-Based Performance Guarantees}\label{sec3}
In this section, we first analyze $|\mathbf{a}_{j}^{*}\mathbf{r}^{i-1}|/d_j$ in step \textbf{1} of Algorithm \ref{alg:omp} to find the conditions on the CS matrix that lead to correct support recovery. Then, assuming that the support of $\mathbf{x}$ is detected correctly, we obtain an upper-bound on the MSE of the vector estimated in step \textbf{3} of the algorithm.
\subsection{Guarantees on successful support recovery}
\par Let $\bar{\mathbf{x}}^{i} = \mathbf{x} - \hat{\mathbf{x}}^{i}$ denote the difference between the original sparse vector $\mathbf{x}$ and its estimated version in the $i^{\mathrm{th}}$ OMP iteration. Hence, in the $i^{\mathrm{th}}$ iteration, the residue $\mathbf{r}^{i}$ is
\begin{equation}\label{eqn:reisd_i}
    \mathbf{r}^{i} = \mathbf{A}\bar{\mathbf{x}}^{i} + \mathbf{v},
\end{equation}
and the support recovery criterion in step \textbf{1} is equivalent to
\begin{equation}\label{eqn:expanded_supp}
    s         =\underset{j\in[N]}{\mathrm{argmax}} ~\dfrac{|\mathbf{a}_{j}^{*}\mathbf{A}\bar{\mathbf{x}}^{i-1} + \mathbf{a}_{j}^{*}\mathbf{v}|}{d_j}.
\end{equation}
As we observe from \eqref{eqn:expanded_supp}, the noise term $\mathbf{a}_{j}^{*}\mathbf{v}/d_j$ can lead to incorrect support detection. Therefore, to control support misdetection under noise, we first examine the event
\begin{equation}\label{eqn:event}
E = \left\lbrace \underset{j\in[N]}{\max}~\frac{|\mathbf{a}_{j}^{*}\mathbf{v}|}{d_j}<\rho \right\rbrace
\end{equation} 
with $\rho \coloneqq \sigma \sqrt{2(1+\alpha)\log{N}}$ and $\alpha\! >\! 0$ \cite{ben2010coherence}. In Lemma \ref{lm:event}, we provide a lower bound on the probability of the event $E$. A similar lemma was derived in \cite{ben2010coherence} for a real-valued CS problem with a column-normalized CS matrix, i.e., $d_{\mathrm{max}} = d_{\mathrm{min}}$. Our result, however, applies to a more general CS setting.
\begin{lemma}\label{lm:event}
    The probability of the event $E$ in \eqref{eqn:event} is lower bounded as
    \begin{equation}\label{eqn:lemma1}
        \mathrm{Pr}\!\left\lbrace E \right\rbrace \geq \left(1-\sqrt{\frac{2}{\pi}}.\sqrt{\frac{\sigma}{\rho}}\exp{\!\left(\!-\frac{\rho^2}{2\sigma^2}\right)}\right)^{2N}.
    \end{equation}
\end{lemma}
\begin{proof}
    We use $\Re(\mathbf{a}_{j}^{*}\mathbf{v})$ and $\Im(\mathbf{a}_{j}^{*}\mathbf{v})$ to denote the real and imaginary parts of $\mathbf{a}_{j}^{*}\mathbf{v}$. Since the entries of $\mathbf{v}$ are independently distributed as $\mathcal{CN}(0,\sigma^2)$, the random variables $\{\Re(\mathbf{a}_{j}^{*}\mathbf{v})/d_j\}$ and $\{\Im(\mathbf{a}_{j}^{*}\mathbf{v})/d_j\}$ $\forall j\in[N]$ are jointly Gaussian and each distributed as $\mathcal{N}(0,\sigma^2/2)$. For any $\kappa \in \mathbb{C}$, we notice that $|\kappa|< \rho$ whenever $\Re(\kappa)<\rho/\sqrt{2}$ and $\Im(\kappa)<\rho/\sqrt{2}$. Using this observation in \eqref{eqn:event}, we can write 
    \begin{align*}\label{eqn:boundE}
    \mathrm{Pr}\!\left\lbrace E \right\rbrace &\geq \mathrm{Pr}\!\left\lbrace \underset{j\in[N]}{\max}\frac{|\Im(\mathbf{a}_{j}^{*}\mathbf{v})|}{d_j}<\!\frac{\rho}{\sqrt{2}}\cap\! \underset{j\in[N]}{\max}\frac{|\Re(\mathbf{a}_{j}^{*}\mathbf{v})|}{d_j}<\!\frac{\rho}{\sqrt{2}}\right\rbrace\\
    &\overset{\mathrm{(a)}}{\geq}\!\prod_{j\in[N]}\!\mathrm{Pr}\!\left\lbrace \frac{|\Im(\mathbf{a}_{j}^{*}\mathbf{v})|}{d_j}<\!\frac{\rho}{\sqrt{2}}\right\rbrace\mathrm{Pr}\!\left\lbrace\frac{|\Re(\mathbf{a}_{j}^{*}\mathbf{v})|}{d_j}<\!\frac{\rho}{\sqrt{2}}\right\rbrace\!.
\end{align*}
In $\mathrm{(a)}$, we use \v{S}id\'ak's lemma \cite[Theorem 1]{vsidak1967rectangular}. Now, replacing the probability values in the right-hand of the above equation with their upper limits similar to \cite{ben2010coherence}, we obtain \eqref{eqn:lemma1}.
\end{proof}
\par Next, we obtain a condition that guarantees the successful recovery of the support of $\mathbf{x}$ using the OMP when the event $E$ in \eqref{eqn:event} occurs. This condition depends on the magnitude of the nonzero entries of the  sparse vector $\mathbf{x}$, the noise variance $\sigma^2$, the $\ell_2$-norm of the columns of the CS matrix, i.e., $d_j$ $\forall j\in [N]$, and the coherence of the CS matrix $\mu$ \eqref{eqn:coherence}.
\begin{theorem}\label{theorem1}
Let $x_{\mathrm{max}} = \underset{j\in\Lambda}{\max}~|x_j|$ and $x_{\mathrm{min}} = \underset{j\in\Lambda}{\min}~|x_j|$. If 
\begin{equation}\label{eqn:cond_min}
    d_{\mathrm{min}}x_{\mathrm{min}} - (2k-1)\mu d_{\mathrm{max}}x_{\mathrm{min}} \geq 2\rho
\end{equation}
and the event $E$ in \eqref{eqn:event} occurs, the OMP algorithm successfully recovers the support of $\mathbf{x}$.
\end{theorem}
\begin{proof}
Our proof follows a similar structure as the one in \cite[Lemma 3]{ben2010coherence}. We begin with the first iteration of Algorithm \ref{alg:omp} and analyze step \textbf{1}. Then, by induction, we show that when the event $E$ occurs and $\mathbf{x}$ satisfies \eqref{eqn:cond_min}, the OMP correctly recovers the support of $\mathbf{x}$ after $k$ iterations.
\par Now, considering the first iteration of Algorithm \ref{alg:omp} and assuming that the event $E$ occurs, we find conditions for which the index of the selected column, i.e., the one that maximizes $|\mathbf{a}_{j}^{*}\mathbf{r}^{0}|/d_j$, belongs to the support $\Lambda$ of $\mathbf{x}$. Noting that $\mathbf{r}^0=\mathbf{y}$, we can equivalently write this condition as
\begin{equation}\label{eqn:detectOMP}
    \underset{j\in\Lambda}{\max}~\frac{|\mathbf{a}_{j}^{*}\mathbf{y}|}{d_j} > \underset{j\notin\Lambda}{\max}~\frac{|\mathbf{a}_{j}^{*}\mathbf{y}|}{d_j}.
\end{equation}
 For the right-hand side of \eqref{eqn:detectOMP}, under the event $E$, we have
\begin{align}\label{eqn:NotInLambda}
    \underset{j\notin\Lambda}{\max}~\frac{|\mathbf{a}_{j}^{*}\mathbf{y}|}{d_{j}} &= \underset{j\notin\Lambda}{\max}~\frac{|\mathbf{a}_{j}^{*}\mathbf{v} + \sum_{i\in\Lambda}x_{i}\mathbf{a}_{j}^{*}\mathbf{a}_{i}|}{d_{j}}\nonumber\\
    &\overset{\mathrm{(b)}}{\leq} \underset{j\notin\Lambda}{\max}~\frac{|\mathbf{a}_{j}^{*}\mathbf{v}|}{d_{j}} + \underset{j\notin\Lambda}{\max}~\frac{\sum_{i\in\Lambda}|x_{i}\mathbf{a}_{j}^{*}\mathbf{a}_{i}|}{d_{j}} \nonumber \\
    &\overset{\mathrm{(c)}}{<} \rho + k\mu d_{\mathrm{max}}x_{\mathrm{max}}.
\end{align}
In $\mathrm{(b)}$, we use the triangular inequality. In $\mathrm{(c)}$, we use the assumption that the event $E$ occurs and the fact that
\begin{equation*}
    \frac{|x_{i}\mathbf{a}_{j}^{*}\mathbf{a}_{i}|}{d_{j}} = |x_{i}|\frac{|\mathbf{a}_{j}^{*}\mathbf{a}_{i}|d_{i}}{d_{j}d_{i}} \overset{\eqref{eqn:coherence}}{\leq} d_{\mathrm{max}}x_{\mathrm{max}}\mu.
\end{equation*}
Under the event $E$, for the left-hand side of \eqref{eqn:detectOMP}, we have
\begin{align}\label{eqn:InLambda}
    \underset{j\in\Lambda}{\max}~\frac{|\mathbf{a}_{j}^{*}\mathbf{y}|}{d_{j}} &= \underset{j\in\Lambda}{\max}~\frac{|d_{j}^{2}x_{j}+\mathbf{a}_{j}^{*}\mathbf{v} + \sum_{i\in\Lambda\backslash\{j\}}x_{i}\mathbf{a}_{j}^{*}\mathbf{a}_{i}|}{d_{j}}\nonumber \\ &\geq d_{\mathrm{min}}x_{\mathrm{max}} - \underset{j\in\Lambda}{\max}~\frac{\left\lvert\mathbf{a}_{j}^{*}\mathbf{v}+\sum_{i\in\Lambda\backslash\{j\}}x_{i}\mathbf{a}_{j}^{*}\mathbf{a}_{i}\right\rvert}{d_{j}} \nonumber \\
    &> d_{\mathrm{min}}x_{\mathrm{max}} - \rho - (k-1)\mu d_{\mathrm{max}}x_{\mathrm{max}}.
\end{align}
Note that in the first inequality, we use $\underset{j\in\Lambda}{\max}~|d_{j}x_{j}|\geq d_{\mathrm{min}}x_{\mathrm{max}}$. Hence, from \eqref{eqn:NotInLambda} and \eqref{eqn:InLambda} we can write 
\begin{equation}\label{eqn:detectOMP2}
    \underset{j\in\Lambda}{\max}\frac{|\mathbf{a}_{j}^{*}\mathbf{y}|}{d_{j}}\! >\! d_{\mathrm{min}}x_{\mathrm{max}} - (2k-1)\mu d_{\mathrm{max}}x_{\mathrm{max}} - 2\rho + \underset{j\notin\Lambda}{\max}\frac{|\mathbf{a}_{j}^{*}\mathbf{y}|}{d_{j}}.
\end{equation}
From \eqref{eqn:detectOMP2}, we observe that under the event $E$, when
\begin{equation}\label{eqn:cond_max}
    d_{\mathrm{min}}x_{\mathrm{max}} - (2k-1)\mu d_{\mathrm{max}}x_{\mathrm{max}} \geq 2\rho,
\end{equation}
equation \eqref{eqn:detectOMP} holds and therefore the selected entry in the first iteration of Algorithm \ref{alg:omp} will belong to the support of $\mathbf{x}$. We note that \eqref{eqn:cond_min} implies \eqref{eqn:cond_max} and use an induction-based technique in the proof of \cite[Theorem. 4]{ben2010coherence}. We then show that under \eqref{eqn:cond_min} and the event $E$, Algorithm \ref{alg:omp} will successfully recover the entire support $\Lambda$ of $\mathbf{x}$ after $k$ iterations. 
\end{proof}
\par The condition in \eqref{eqn:cond_min} gives a lower bound on the weakest coefficient of $\mathbf{x}$ such that its support can be successfully identified with the OMP. This bound applies to general CS matrices whose columns may have unequal norms. For the special case when $d_{\mathrm{max}}=d_{\mathrm{min}}=1$, the bound in \eqref{eqn:cond_min} is exactly the same as the one derived in \cite[Theorem 4]{ben2010coherence}. 
\par The OMP can identify weaker coefficients when using CS matrices that have a small $d_{\mathrm{max}}/d_{\mathrm{min}}$. This observation follows by rewriting \eqref{eqn:cond_min} as 
\begin{equation}
    x_{\mathrm{min}} {\geq} \frac{2\rho}{d_{\mathrm{min}} - (2k-1)\mu d_{\mathrm{max}}} = \frac{2\rho/d_{\mathrm{min}}}{1 - \mu(2k-1) (d_{\mathrm{max}}/d_{\mathrm{min}})}.
\end{equation}
We observe from \eqref{eqn:dmin} and \eqref{eqn:dmax} that the smallest possible $d_{\mathrm{max}}/d_{\mathrm{min}}$ is $1$. 
\begin{remark}\label{remark1}
For the noiseless CS problem, i.e., $\sigma = 0$, Algorithm \ref{alg:omp} correctly recovers the support $\Lambda$ of $\mathbf{x}$ if
\begin{equation}\label{eqn:noiseless}
    k\leq \frac{1}{2}\left(1+\frac{d_{\mathrm{min}}}{ \mu d_{\mathrm{max}}}\right).
\end{equation}
When $d_{\mathrm{max}}\!=d_{\mathrm{min}}$, \eqref{eqn:noiseless} reduces to the one derived in \cite{ben2010coherence}.
\end{remark}
\begin{proof}
    The proof follows from \eqref{eqn:cond_min} by setting $\sigma=0$.
\end{proof}
\par Next, by assuming that the conditions in \eqref{eqn:cond_min} and \eqref{eqn:lemma1} are satisfied, we obtain an upper bound on the MSE of the  vector reconstructed in step \textbf{3} of Algorithm \ref{alg:omp}.
\subsection{Guarantees on the robustness of the reconstruction to noise}
To derive an upper bound on the MSE of the estimated vector in step \textbf{3} of Algorithm \ref{alg:omp}, we first examine the eigenvalues of $(\mathbf{A}_{\Lambda}^{*}\mathbf{A}_{\Lambda})^{-1}$ in Lemma \ref{lemma2}.
\begin{lemma}\label{lemma2}
    Let $\lambda_{\mathrm{max}}\left((\mathbf{A}_{\Lambda}^{*}\mathbf{A}_{\Lambda})^{-1}\right)$ denote the largest eigenvalue of the positive semidefinite matrix $(\mathbf{A}_{\Lambda}^{*}\mathbf{A}_{\Lambda})^{-1}$. Then, 
    \begin{equation}\label{eqn:lemma2}
        \lambda_{\mathrm{max}}\left((\mathbf{A}_{\Lambda}^{*}\mathbf{A}_{\Lambda})^{-1}\right) \leq \frac{1}{d_{\mathrm{min}}\left(d_{\mathrm{min}} - (k-1)\mu d_{\mathrm{max}}\right)}.
    \end{equation}
\end{lemma}
\begin{proof}
    If $\lambda_{\mathrm{min}}\left(\mathbf{A}_{\Lambda}^{*}\mathbf{A}_{\Lambda}\right)$ is the smallest eigenvalue of $\mathbf{A}_{\Lambda}^{*}\mathbf{A}_{\Lambda}$, we can write
    \begin{equation}\label{eqn:eig_maxmin}
        \lambda_{\mathrm{max}}\left((\mathbf{A}_{\Lambda}^{*}\mathbf{A}_{\Lambda})^{-1}\right) = \left(\lambda_{\mathrm{min}}\left(\mathbf{A}_{\Lambda}^{*}\mathbf{A}_{\Lambda}\right)\right)^{-1}.
    \end{equation}
    Now, using the Gershgorin circle theorem \cite[Theorem 7.2.1]{golub2013matrix}, we can bound the eigenvalues $\lambda_i,~ \forall i\in[k]$ of $\mathbf{A}_{\Lambda}^{*}\mathbf{A}_{\Lambda}$ as
\begin{equation}
    \left|\lambda_i-\|\mathbf{a}_i\|_2^2\right|\leq \sum\limits_{\{j\in\Lambda,j\neq i\}} |\mathbf{a}_{i}^{*}\mathbf{a}_j|.
\end{equation}
Therefore, we can write
\begin{align}\label{eqn:eig_min}
    \lambda_{\mathrm{min}}\left(\mathbf{A}_{\Lambda}^{*}\mathbf{A}_{\Lambda}\right) &\geq \underset{i\in\Lambda}{\min}\{d_{i}^{2}-\sum\limits_{\{j\in\Lambda,j\neq i\}} |\mathbf{a}_{i}^{*}\mathbf{a}_j|\}\nonumber\\
    &\overset{\eqref{eqn:coherence}}{\geq} \underset{i\in\Lambda}{\min}~ d_{i}\left(d_{i}-\sum\limits_{\{j\in\Lambda,j\neq i\}} d_j\mu\right)\nonumber\\
    &\geq d_{\mathrm{min}}\left(d_{\mathrm{min}}-(k-1)d_{\mathrm{max}}\mu\right).
\end{align}
Finally, combining \eqref{eqn:eig_min} and \eqref{eqn:eig_maxmin} completes the proof.
\end{proof}
Now, using Lemma \ref{lemma2}, we obtain an upper-bound for the MSE of the vector recovered with Algorithm \ref{alg:omp}.
\begin{theorem}\label{theorem2}
    Let $\hat{\mathbf{x}}$ denote the vector reconstructed using the OMP after $k$ iterations. If \eqref{eqn:lemma1} and \eqref{eqn:cond_min} hold, we can write
    \begin{equation}\label{eqn:theorem2}
        \|\hat{\mathbf{x}}-\mathbf{x}\|_2^2 \leq \left(\frac{d_{\mathrm{max}}}{d_{\mathrm{min}}}\right)^{2}\frac{k\rho^2}{\left(d_{\mathrm{min}} - (k-1)\mu d_{\mathrm{max}}\right)^2}.
    \end{equation}
\end{theorem}
\begin{proof}
    Since both $\mathbf{x}$ and $\hat{\mathbf{x}}$ are supported on $\Lambda$, we have
    \begin{align}\label{eqn:theorem2_prf}
        \|\hat{\mathbf{x}}-\mathbf{x}\|_2^2 &=\! \|(\mathbf{A}_{\Lambda}^{*}\mathbf{A}_{\Lambda})^{-1}\mathbf{A}_{\Lambda}^{*}\mathbf{y}-\mathbf{x}_{\Lambda}\|^2\nonumber
        =\! \|(\mathbf{A}_{\Lambda}^{*}\mathbf{A}_{\Lambda})^{-1}\mathbf{A}_{\Lambda}^{*}\mathbf{v}\|^2\!\nonumber\\
        &\leq \left[\lambda_{\mathrm{max}}\!\left((\mathbf{A}_{\Lambda}^{*}\mathbf{A}_{\Lambda})^{-1}\right)\right]^2\sum \mathop{}_{\mkern-5mu j\in\Lambda}|\mathbf{a}_{j}^{*}\mathbf{v}|^2 \nonumber\\
        &\leq \left[\lambda_{\mathrm{max}}\!\left((\mathbf{A}_{\Lambda}^{*}\mathbf{A}_{\Lambda})^{-1}\right)\right]^2 k(\, \mathrm{max}_j |\mathbf{a}_{j}^{*}\mathbf{v}|)^2.
    \end{align}
    We use \eqref{eqn:event} and \eqref{eqn:lemma2} in \eqref{eqn:theorem2_prf} to get \eqref{eqn:theorem2}. Finally, we notice that the bound in \eqref{eqn:theorem2} is tight when  $d_{\mathrm{max}}/d_{\mathrm{min}}$ approaches $1$.
\end{proof}
\vspace{-2,5mm}
\section{Simulations}\label{sec5}
We discuss our work in the context of CS-based sparse spatial channel estimation between a phased array comprising $N=32$ antennas and a single antenna receiver. The sparse wireless channel in this setup is modeled as an $N\times 1$ complex vector $\mathbf{x}$ with $k=2$ non-zero entries. The $k$ indices are chosen uniformly at random and the corresponding non-zero coefficients are sampled from $\mathcal{CN}(0,1)$ distribution. Then, the sparse vector is normalized to a unit norm.
\par We define $\mathbb{Q}=\{q: |q|=1\}$ to model the weights applied in a phased array. The channel measurements in our setting are acquired by applying $M<N$ distinct random circulant shifts of a codeword $\mathbf{f}\in \mathbb{Q}^N$ to a phased array. We use $\mathbf{F}\in \mathbb{Q}^{M \times N}$ to denote the matrix comprising the applied codewords, i.e., every row of $\mathbf{F}$ is a circularly shifted version of $\mathbf{f}$. This circulant shift-based measurement technique was discussed in \cite{myers2019falp,myers2018spatial} for channel estimation at millimeter wave frequencies. The CS matrix in this case can be expressed as $\mathbf{A} = \mathbf{F}\mathbf{U}_N$, where $\mathbf{U}_N$ is an $N\times N$ unitary discrete Fourier transform (DFT) matrix \cite{myers2018spatial}.
\par We describe CS matrices generated by selecting different vectors for $\mathbf{f}$. It can be shown that the CS matrix $\mathbf{A}$ has equal column norms when the magnitude of the DFT of $\mathbf{f}$ is flat. A Zadoff-Chu (ZC) sequence $\mathbf{f}^{\mathrm{ZC}}$ \cite{chu1972polyphase}, defined by $f^{\mathrm{ZC}}_i = \exp{\left({-\mathsf{j}\pi (i-1)^2}/{N}\right)},~ i\in[N]$, is an example of a sequence that achieves a flat DFT magnitude. Implementing this sequence, however, requires $\log_2{N}$-bit phase shifters for an $N$-element array; this corresponds to $5$-bit phase shifters in our setup. In practice, low-resolution phase shifters (e.g. $2$-bit) are preferred due to their low power and low hardware complexity \cite{anjos201914}. Under this low-resolution constraint, it is not always possible to construct an $\mathbf{f}$ that has a uniform DFT magnitude \cite{armario2020almost}. In such a case, $\mathbf{f}$ may be chosen at random from the set of feasible low-resolution codes for CS \cite{romberg2009compressive}. 
\par In this paper, we chose two codes $\mathbf{f}^{\mathrm{1}}$ and $\mathbf{f}^{\mathrm{2}}$ from the $2$-bit alphabet $\{1,\imj, -1, -\imj\}^{32}$. The vector $\mathbf{f}^{\mathrm{1}}$ resulted in CS matrices with  $d_\mathrm{max}/d_\mathrm{min}\approx 2.43$, and $\mathbf{f^{\mathrm{2}}}$ resulted in CS matrices with $d_\mathrm{max}/d_\mathrm{min}\approx 2.83$. For a benchmark, we also consider the ZC sequence over a $5$-bit alphabet that results in CS matrices with  $d_\mathrm{max}/d_\mathrm{min}=1$. We use $\rho = 2.63\sigma$ in \eqref{eqn:theorem2}, which corresponds to $\mathrm{Pr}\!\left\lbrace E \right\rbrace \geq 0.97$ in \eqref{eqn:lemma1}. The three CS matrix designs are evaluated with the OMP in terms of the probability of successfully identifying the support and the normalized MSE (NMSE) of the estimate. The NMSE of the estimated sparse vector $\hat{\mathbf{x}}$ is defined as $\mathbb{E}[\|\mathbf{x}-\hat{\mathbf{x}}\|_2^2]/\mathbb{E}[\|\mathbf{x}\|_2^2]$. 
\par We observe from Fig. \ref{fig:nmse} that CS matrices with $d_\mathrm{max}/d_\mathrm{min}=1$ result in about $1.43$ dB lower NMSE than those with $d_\mathrm{max}/d_\mathrm{min}\in\{2.43,2.83\}$. The observation aligns with our findings in Theorem \ref{theorem2}, as shown by the upper-bound plots in Figure~\ref{fig:nmse}. The probability of support recovery is defined as the fraction of instances when the OMP exactly recovers $\Lambda$. Fig. \ref{fig:prob} shows that CS matrices with a smaller $d_{\mathbf{max}}/d_{\mathbf{min}}$ lead to a higher probability in recovering the sparse support. To achieve $98\%$ success in support recovery, we observe that $\mathbf{x}_{\mathrm{min}} \approx 0.79 \sigma $ when $d_\mathrm{max}/d_\mathrm{min}=1$, and $\mathbf{x}_{\mathrm{min}} \approx 1.05 \sigma$ when $d_\mathrm{max}/d_\mathrm{min}=2.43$. Thus, CS matrices with a small $d_\mathrm{max}/d_\mathrm{min}$ are desirable for the OMP as they allow support recovery even when the sparse coefficients are weak in magnitude. 
\begin{figure}[h]
\centering
\includegraphics[trim=7.3cm 7.1cm 7.3cm 7.8cm, width=0.14\textwidth]{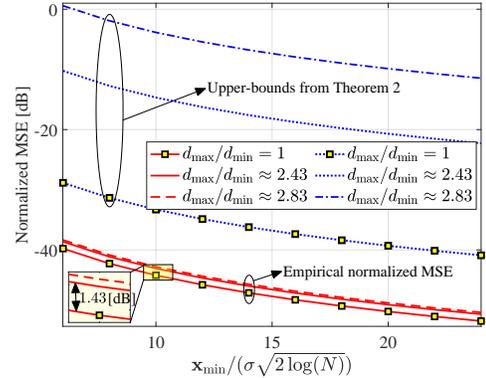}
\caption{\small 
CS matrices with equal column norms, i.e., $d_\mathrm{max}/d_{\mathrm{min}}=1$, result in smaller normalized MSE with the OMP than those with different column norms. This result coincides with our analysis in Theorem \ref{theorem2}. 
In this example, $N=32$, $M=20$ and $k=2$.\normalsize}\label{fig:nmse} 
\end{figure}
\begin{figure}[h]
\centering
\includegraphics[trim=7.3cm 7.1cm 7.3cm 8.5cm, width=0.14\textwidth]{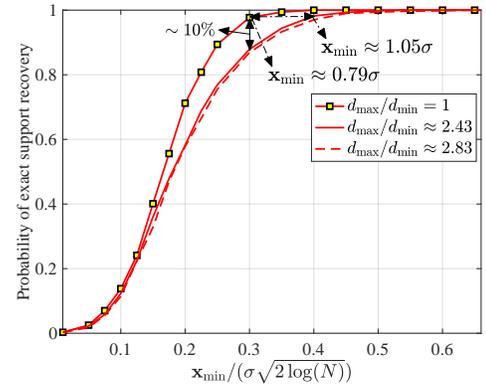}
\caption{\small For CS via the OMP, CS matrices with equal column norms result in a higher support recovery probability than those with different column norms. Here, $N=32$, $M=20$ and $k=2$.\normalsize}\label{fig:prob} 
\end{figure}
\vspace{-3,6mm}
\section{Conclusions}\label{sec6conclusion}
The OMP algorithm is often analyzed under the assumption of equal column norms in CS matrices, but this assumption is not always true in typical CS applications. In this paper, we derived performance guarantees for the OMP when the CS matrix has unequal column norms. Our results suggest that OMP performs better with CS matrices whose ratio of the maximum to minimum column norm is close to $1$. This conclusion is important in hardware-constrained applications where finding CS matrices with equal norms may not be feasible, and it can guide the selection of effective CS matrices. 

\ifCLASSOPTIONcaptionsoff
  \newpage
\fi

\bibliography{References}
\bibliographystyle{ieeetr}

\end{document}